\documentclass[11pt]{article}

\usepackage[T1]{fontenc}
\usepackage[utf8]{inputenc}
\usepackage{amsmath,amssymb,amsthm,mathtools}
\usepackage[margin=1in]{geometry}
\usepackage{enumitem}
\usepackage{xcolor}
\usepackage[hidelinks]{hyperref}

\newtheorem{theorem}{Theorem}
\newtheorem{lemma}{Lemma}

\theoremstyle{remark}
\newtheorem*{remark}{Remark}

\newcommand{\AD}{\mathrm{AD}}
\newcommand{\Reach}{\operatorname{Reach}}
\newcommand{\1}{\mathbf{1}}

\title{The Local-to-Global $\AD$-$k$ Conjecture is Resolved}
\author{Wei Chen \footnote{The proof was found with the help of GPT-5.6 Sol.}\\
		Microsoft Research Asia \\
		\tt{weic@microsoft.com}}
\date{}

\begin{document}

\maketitle


\begin{abstract}
$\AD$-$k$ stands for {\em Alternating Differences through order \(k\)}, and it is a property of set functions denoting that
	the first order difference of the set function is nonnegative (a.k.a. monotnocity), the second order difference is nonpositive (a.k.a. submodularity), and so on
	with signs alternating through order $k$.
Chen et al. \cite{adkConjecture} conjectured that in an influence diffusion model called the general threshold model 
	originally defined by Kempe et al.~\cite{kempe2003infmax}, 
	if every local influence function is $\AD$-$k$, then the global influence spread function is also $\AD$-$k$, for any (possibly cyclic)
	directed graph and any $k$.
This paper provides a complete proof showing that the conjecture is true.
The proof utilizes M\"obius inversion, reverse reachable sets, and decision tree partition techniques and extends the probability distribution of node triggering sets into
	a generalized algebraic structure allowing negative weights for triggering sets.
The extension to negative-weighted triggering sets may be of independent interest and may have further algorithmic implications.

%
%
%
\end{abstract}

\section{Introduction}
\label{sec:introduction}

Influence maximization is the task of finding a small set of seed nodes in a
social network so that the influence spread produced by these seeds under a
specified diffusion model is maximized~\cite{kempe2003infmax}.  
The influence-maximization problem grew out of earlier work on exploiting
network effects in marketing and information propagation
\cite{domingos2001mining,richardson2002mining}, and it has broad applications, 
	including viral marketing, outbreak detection, and
	limiting the spread of misinformation
	\cite{richardson2002mining,Leskovec2007costeffective,Budak2011misinformation}.
Influence diffusion modeling and
	influence maximization tasks have been extensively studied (cf. \cite{chen2013information,Li2018survey}).
Among these studies, a substantial algorithmic literature has developed scalable methods for
influence maximization, ranging from efficient greedy and heuristic
implementations \cite{ChenWY09efficientinfluence,chen2010sharpphard} to
	reverse-reachability sampling algorithms with near-linear running times
	\cite{borgs2014rrset,tang2014newrrset,tang2015rrset}, the latter of which
	have intrinsic connection with the prove method of this paper.

Underlying this optimization problem is a stochastic diffusion model that
characterizes how influence propagates through a network from the seed nodes.
Kempe et al. formulated a number of classical influence diffusion models,
	including the independent cascade model, the linear threshold model,
	the triggering model, and the general threshold model
	as the generalization of all the above models~\cite{kempe2003infmax}.
Consider a social network $G=(V,E)$ represented as a directed graph, where
$V$ is the set of nodes and $E$ is the set of directed edges.  In the
{\em general threshold model}, every node $v$ has a local influence function
$f_v:2^{N_v}\rightarrow[0,1]$, where $N_v$ is the set of in-neighbors of $v$.
The diffusion proceeds in discrete steps $\tau=0,1,2,\ldots$.  At step
$\tau=0$, all nodes in a given seed set $X$ are activated, while all other
nodes are inactive.  Independently, every node $v\in V$ samples a threshold
$\theta_v\in[0,1]$ uniformly at random.  At every step $\tau\geq1$, let
$S_{v,\tau-1}$ be the set of in-neighbors of an inactive node $v$ that are
active by step $\tau-1$.  The node $v$ is activated at step $\tau$ if
$f_v(S_{v,\tau-1})\geq\theta_v$; otherwise, it remains inactive.  The process
continues until no additional node is activated.  The global {\em influence
spread} function $\sigma(X)$ is the expected number of nodes active at the end
of the diffusion process starting from seed set $X$.

In their seminal paper, Kempe et al.~\cite{kempe2003infmax} conjectured that if
the local influence functions $f_v$ are monotone and submodular for all
$v\in V$, then the global influence spread function $\sigma$ is also monotone
and submodular.  
This conjecture captures a fundamental local-to-global
property of the general threshold model.  Mossel and Roch later proved the
conjecture using a reverse coupling argument~\cite{Mossel2010local2global}.
Global submodularity is algorithmically important because the classical greedy
algorithm gives a $(1-1/e)$-approximation for maximizing a nonnegative monotone
submodular function subject to a cardinality constraint
\cite{Nemhauser1978submodular}.  Thus, a local-to-global preservation theorem
connects structural assumptions on the diffusion process with approximation
guarantees for selecting seed nodes.

Chen et al.~\cite{adkConjecture} observed that the triggering model has a
stronger alternating-difference property: first-order differences are
nonnegative, second-order differences are nonpositive, third-order differences
are nonnegative, and the signs continue to alternate at higher orders.  They
therefore introduced the notion of {\em Alternating Difference-$k$}, or
$\AD$-$k$, which requires this alternating-sign condition through order $k$.
For example, the first-order difference is
\[
  f(S\cup\{v\})-f(S),
\]
and the second-order difference is
\[
  f(S\cup\{u,v\})-f(S\cup\{u\})-f(S\cup\{v\})+f(S).
\]
They conjectured that, for every directed graph and every relevant order $k$,
if all local influence functions are $\AD$-$k$, then the global influence
spread function is also $\AD$-$k$.  The $\AD$-$1$ property is monotonicity,
and $\AD$-$2$ is equivalent to monotonicity together with submodularity, so
the result of Mossel and Roch settles the cases $k=1,2$.  Chen et al. also
proved the conjecture for directed acyclic graphs (DAGs) at every order and for
general directed graphs in the $\AD$-$\infty$ case, where the alternating-sign
condition holds at all possible orders.  Thus, the remaining open cases were
cyclic graphs and finite orders $3\leq k<|V|$.  To the best of our knowledge,
these cases had not previously been resolved.

In this paper, we prove the full conjecture for every $k$ and every finite
directed graph.  
The proof utilizes M\"obius inversion, reverse reachable sets via triggering sets, and decision tree partition techniques, and
	it crucially extends the probability distribution of node triggering sets into
	a generalized algebraic structure allowing negative weights for triggering sets.
The reverse reachable sets were originally proposed as an important algorithmic technique for achieving near linear time scalable influence 
	maximization algorithms for the triggering model, a subclass of the general threshold model~\cite{borgs2014rrset,tang2014newrrset,tang2015rrset}.
Therefore, beyond closing the local-to-global $\AD$-$k$ conjecture,
	the extension to negative-weighted triggering sets may be of independent interest and may have further algorithmic implications for the general threshold model.

From a broader set-function perspective, these alternating-sign conditions are
instances of higher-order monotonicity for pseudo-Boolean functions
\cite{Sethpan2005booleanfunction}.  M\"obius transformations provide standard
coordinate representations of set functions and their higher-order
interactions \cite{Grabisch2000mobius}.  Related preservation results for
higher-order alternating functions under compounding were subsequently
developed by Ressel \cite{Ressel2023compounding}, who
	presented a shorter and more general proof of the compounding theorem that forms a principal analytic 
	ingredient in \cite{adkConjecture} for DAGs. 
His result does not address the feedback dependencies arising from directed 
	cycles and therefore does not resolve general local-to-global $\AD$-$k$ conjecture for general cyclic graphs and finite \(k\).

The paper is organized as follows.
Section~\ref{sec:model} formally defines the general threshold model and the
$\AD$-$k$ property and states the main local-to-global theorem.  Section~\ref{sec:signed-coordinates}
uses M\"obius inversion to represent each local influence function by a
normalized signed measure over triggering sets, and Section~\ref{sec:signed-cycles}
proves that this signed triggering representation remains algebraically
equivalent to the live-edge graph reachability even on graphs with cycles.
Section~\ref{sec:cylinders} expresses each alternating difference of a
target-activation probability as the signed mass of a corresponding live-edge
reachability cylinder.  Section~\ref{sec:decision-tree} constructs a two-phase
edge-query decision tree that partitions this cylinder into disjoint classes
of triggering configurations in which at most $|A|$ incoming edges at any node
are required to be present.  Section~\ref{sec:main-proof} combines
this partition with the local $\AD$-$k$ interval inequalities to prove the
theorem.
We conclude the paper with Section~\ref{sec:conclude}.

\section{Model and Statement}
\label{sec:model}

Let $G=(V,E)$ be a finite directed graph with $V$ as the set of vertices or nodes and $E$ as the set of directed edges.  
The {\em general threshold model} is specified as follows.
Every vertex $v\in V$ has a local function
\[
  f_v:2^{N_v}\longrightarrow[0,1],
  \qquad f_v(\varnothing)=0,
\]
where $N_v$ is its set of in-neighbors.  The function is monotone.
Independently draw $\theta_v$ uniformly from $[0,1]$.  Starting from a seed
set $X$, an inactive vertex $v$ becomes active once
\[
  f_v(\text{active in-neighbors of }v)\geq \theta_v.
\]
Activation is progressive.  Let $p_t(X)$ be the probability that target $t$
is eventually active, and let
\[
  \sigma(X)=\sum_{t\in V}p_t(X),
\]
where $\sigma(X)$ is called the {\em influence spread} of seed set $X$.

The {\em triggering model} \cite{kempe2003infmax} is a subclass of the general threshold model, and it is conceptually related to the proof of the paper.
In the triggering model, every node $v\in V$ has a random triggering set $T_v\subseteq N_v$ sampled from $v$'s in-neighbors. 
All triggering sets $(T_v)_{v\in V}$ collected form a {\em live-edge graph} $L = (V, E_L)$, where $E_L = \{(u,v) | u\in T_v, v\in V\}$.
Diffusion from the seed set $X$ follows the reachability relation on the live-edge graph, that is, all nodes that can be reached from $X$ in the live-edge graph $L$ are activated.
It is easy to verify that if for every $v$, we set $f_v = \sum_{T_v \subseteq N_v, T_v\cap S \ne \emptyset} D_v(T_v)$, where $D_v(T_v)$ is the probability that $T_v$ is sampled, then
	we can show that the trigger model is a special case of the general threshold model.

For disjoint $A,S\subseteq V$, define the set difference
\begin{equation}
  \Delta_A g(S)
  =\sum_{B\subseteq A}(-1)^{|A|-|B|}g(S\cup B).
  \label{eq:difference}
\end{equation}
A set function $g$ is $\AD$-$k$ if
\begin{equation}
  (-1)^{|A|+1}\Delta_Ag(S)\geq0
  \label{eq:adk}
\end{equation}
whenever $1\leq |A|\leq k$ and $A\cap S=\varnothing$.

The main result we aim to prove is the follow theorem:
\begin{theorem}[Local-to-global $\AD$-$k$]
\label{thm:main}
If every local function $f_v$ is $\AD$-$k$, then every target activation
probability $p_t$ and the influence spread $\sigma$ are $\AD$-$k$.
\end{theorem}


\section{M\"obius Inversion and Signed Triggering Weights}
\label{sec:signed-coordinates}

Fix $v$ and abbreviate $N=N_v$ and $f=f_v$.  Define a signed measure $q$ on
subsets of $N$ by
\begin{equation}
  q(\varnothing)=1-f(N),
  \qquad
  q(T)=(-1)^{|T|+1}\Delta_Tf(N\setminus T)
  \quad(T\neq\varnothing).
  \label{eq:qdef}
\end{equation}
Notation $[A,B]$ with $A\subseteq B \subseteq V$ denotes $\{S\subseteq V | A\subseteq S \subseteq B\}$ and is called
a {\em Boolean interval}.

\begin{lemma}[Boolean Interval Masses]
\label{lem:interval}
The signed measure $q$ satisfies
\begin{equation}
  \sum_{T\subseteq N}q(T)=1,
  \qquad
  f(X)=\sum_{T:\,T\cap X\neq\varnothing}q(T).
  \label{eq:hitting}
\end{equation}
For every $L\subseteq U\subseteq N$,
\begin{equation}
  q([L,U])
  :=\sum_{L\subseteq T\subseteq U}q(T)
  =
  \begin{cases}
    1-f(N\setminus U),
      &L=\varnothing,\\[2mm]
    (-1)^{|L|+1}\Delta_Lf(N\setminus U),
      &L\neq\varnothing.
  \end{cases}
  \label{eq:intervalmass}
\end{equation}
\end{lemma}

\begin{proof}
Define
\[
  r(U):=1-f(N\setminus U),
  \qquad U\subseteq N.
\]
We first show that \eqref{eq:qdef} defines $q$ as the M\"obius transform of
$r$ on the Boolean lattice $2^N$.  For $T=\varnothing$, this follows directly
from
\[
  r(\varnothing)=1-f(N)=q(\varnothing).
\]
Now let $T\neq\varnothing$.  Then
\begin{align*}
  \sum_{C\subseteq T}(-1)^{|T|-|C|}r(C)
  & = \sum_{C\subseteq T}(-1)^{|T|-|C|}
      -\sum_{C\subseteq T}(-1)^{|T|-|C|}f(N\setminus C)\\
  & = -\sum_{C\subseteq T}(-1)^{|T|-|C|}f(N\setminus C).
\end{align*}
Make the change of variables $B=T\setminus C$.  Since
\[
  |T|-|C|=|B|
  \quad\text{and}\quad
  N\setminus C=(N\setminus T)\cup B,
\]
we obtain
\[
  \sum_{C\subseteq T}(-1)^{|T|-|C|}r(C)
  =-\sum_{B\subseteq T}(-1)^{|B|}
       f((N\setminus T)\cup B).
\]
On the other hand, the definition of the difference operator gives
\begin{align*}
  (-1)^{|T|+1}\Delta_Tf(N\setminus T)
  & =(-1)^{|T|+1}
      \sum_{B\subseteq T}(-1)^{|T|-|B|}
       f((N\setminus T)\cup B)\\
  & =-\sum_{B\subseteq T}(-1)^{|B|}
       f((N\setminus T)\cup B).
\end{align*}
Thus \eqref{eq:qdef} is equivalently
\[
  q(T)=\sum_{C\subseteq T}(-1)^{|T|-|C|}r(C)
  \qquad\text{for every }T\subseteq N.
\]

The Boolean-lattice M\"obius inversion formula (see, e.g.,
\cite{Grabisch2000mobius}) now gives
\[
  r(U)=\sum_{T\subseteq U}q(T).
\]

Taking $U=N$ and using $f(\varnothing)=0$ gives
\[
  \sum_{T\subseteq N}q(T)=r(N)=1.
\]
Moreover,
\begin{align*}
  \sum_{T:\,T\cap X\neq\varnothing}q(T)
  & =r(N)-r(N\setminus X)\\
  & =1-[1-f(X)]
   =f(X),
\end{align*}
which proves \eqref{eq:hitting}.  Finally, inclusion--exclusion for the
requirement $L\subseteq T$, together with the formula for $r$, gives
\eqref{eq:intervalmass}.
\end{proof}

If $f$ is $\AD$-$k$, Lemma~\ref{lem:interval} has the immediate consequence
\begin{equation}
  q([L,U])\geq0
  \qquad\text{whenever }L\subseteq U\subseteq N\text{ and }|L|\leq k.
  \label{eq:positiveinterval}
\end{equation}
For nonempty $L$, this is precisely a local $\AD$-$k$ inequality at base
$N\setminus U$.  For empty $L$, it follows from $f\leq1$.  Individual atoms
$q(T)$ of rank greater than $k$ may still be negative, and $q$ is not being
treated as a probability distribution.

\section{Signed Triggering and Live-edge Graph Reachability}
\label{sec:signed-cycles}

With the result in the previous section, we can set up the reverse sampling view of the diffusion, in the same spirit as
	originally proposed by \cite{borgs2014rrset}.
For a trigger choice $T_v\subseteq N_v$, put a live edge $u\to v$ exactly when
	$u\in T_v$.  
Collectively, $(T_v)_{v\in V}$ provide a live-edge graph of the diffusion.
For a given target $t \in V$, the set of nodes reachable to $t$ in the live-edge graph is the reverse reachable set $R_t$ of $t$.
In the original Triggering model of \cite{kempe2003infmax}, every triggering set $T_v$ has a non-negative probability and these probabilities sum to $1$
	for each $v$.
With the result in the previous section, in particular Lemma~\ref{lem:interval}, we extend this algebraically such that the atomic $q(T_v)$'s admit negative weights so 
	they are no longer probabilities, but they still sum to $1$.
The key of our analysis is to continue the live-edge graph and reverse reachability argument with this algebraic extension, despite the negative weights at the
	atomic level.

For a seed set $X$ and a threshold realization
$\boldsymbol{\theta}=(\theta_v)_{v\in V}$, let
$A_\infty(X;\boldsymbol{\theta})\subseteq V$ denote the set of nodes active
when the diffusion terminates.  A {\em terminal-state statistic} is any
function
\[
  H:2^V\longrightarrow\mathbb{R}
\]
whose value depends only on this final active set, rather than on the order or
times at which nodes were activated.  Examples include the final number of
active nodes, $H(Y)=|Y|$, and the indicator that a fixed target $t$ is active,
$H_t(Y)=\1\{t\in Y\}$.  For a triggering-set configuration
$\boldsymbol{T}=(T_v)_{v\in V}$, let $L_{\boldsymbol{T}}(X)$ denote the
live-edge closure of $X$, namely, the set of all nodes reachable from $X$ in
the live-edge graph induced by $\boldsymbol{T}$.

\begin{lemma}[Signed Triggering Simulation]
\label{lem:simulation}
For every seed set $X$ and every terminal-state statistic $H$,
\begin{equation}
  \mathbb{E}_{\boldsymbol{\theta}}
  \bigl[H(A_\infty(X;\boldsymbol{\theta}))\bigr]
  =\sum_{\boldsymbol{T}=(T_v)_{v\in V}}
     \left(\prod_{v\in V}q_v(T_v)\right)
     H(L_{\boldsymbol{T}}(X)).
  \label{eq:simulation}
\end{equation}
This holds on arbitrary directed graphs, including graphs with cycles.
\end{lemma}

\begin{proof}
We expose the diffusion using a fixed fair procedure that repeatedly scans the
inactive nodes.  Whenever the set of active in-neighbors of an inactive node
$v$ has grown, the procedure asks whether $v$ should activate.  Because the
diffusion is progressive and every $f_v$ is monotone, this exposure procedure
produces the same terminal active set as the usual synchronous process: both
compute the least stable active set containing $X$.

Here a transcript records the sequence of activation tests, together with the
node tested, its current set of active in-neighbors, and the yes/no outcome of
the test; it is complete when the exposure procedure has terminated.  For each
nonseed node $v$, the distinct active in-neighbor sets appearing in a complete
transcript before $v$ activates form a nested chain
\[
  C_0\subset C_1\subset\cdots\subset C_m.
\]
For notational convenience, put $C_{-1}=\varnothing$ and recall that
$f_v(\varnothing)=0$.  Consider a complete transcript in which the activation
test for $v$ fails at $C_0,\ldots,C_{i-1}$ and first succeeds at $C_i$
(where $i=0$ is allowed).  In the threshold model this means
\[
  f_v(C_{i-1})<\theta_v\leq f_v(C_i),
\]
so the probability of this local answer pattern is
\begin{equation}
  f_v(C_i)-f_v(C_{i-1}).
  \label{eq:thresholdtranscript}
\end{equation}
If $v$ never activates, the corresponding probability is
$1-f_v(C_m)$.  Repeated query sets may be deleted, since they reveal no new
answer.

Now replace the threshold $\theta_v$ by a triggering set $T_v$ with signed
weights $q_v$.  The test at $C$ succeeds exactly when
$T_v\cap C\neq\varnothing$.  Because the sets $C_j$ are nested, the same local
answer pattern occurs exactly when $T_v$ misses $C_{i-1}$ but intersects
$C_i$.  By \eqref{eq:hitting}, its signed mass is
\begin{equation}
  [1-f_v(C_{i-1})]-[1-f_v(C_i)]
  =f_v(C_i)-f_v(C_{i-1}),
  \label{eq:triggertranscript}
\end{equation}
which is identical to \eqref{eq:thresholdtranscript}.  Likewise, the signed
mass of triggering sets that miss $C_m$ is $1-f_v(C_m)$, matching the event
that $v$ never activates.

Although the sets queried at one node may depend on answers previously
obtained at other nodes, this adaptivity causes no problem.  Once a complete
transcript is fixed, the query chain and answer pattern at every node are
fixed.  Independence of the thresholds makes the probability of that
transcript the product of its local factors.  The product signed measure
$\prod_v q_v$ assigns the same product of local factors to the corresponding
trigger transcript by \eqref{eq:thresholdtranscript} and
\eqref{eq:triggertranscript}.  A seeded or otherwise never-queried node
contributes the factor one in both models: ordinary probability one versus
$\sum_{T_v\subseteq N_v}q_v(T_v)=1$.

Finally, sum over all complete transcripts.  Under the trigger interpretation,
a node becomes active precisely when it is reachable from $X$ by live edges,
so the terminal active set is $L_{\boldsymbol{T}}(X)$.  This proves
\eqref{eq:simulation}.  Cycles do not affect the argument: progressiveness
still makes every local query sequence nested and ensures that the exposure
procedure terminates after finitely many activations.  The proof uses only
finite algebraic sums and never assumes that the individual values
$q_v(T_v)$ are nonnegative.
\end{proof}

\begin{remark}[Polynomial-continuation check]
There is a second algebraic justification.  Terminal-set probabilities in the
threshold model are polynomials in the local table entries $f_v(Y)$.  Signed
live-edge probabilities are also polynomials, because $f_v$ and $q_v$ are
related by an invertible affine transformation.  The two polynomials agree on
the full-dimensional open set of strictly positive genuine trigger
distributions; hence they agree identically.  This gives an alternative proof
of Lemma~\ref{lem:simulation}.
\end{remark}

We now specialize Lemma~\ref{lem:simulation} to the terminal-state statistic
$H_t(Y)=\1\{t\in Y\}$.  Its threshold-model expectation is exactly the target
activation probability $p_t(X)$.  For a triggering configuration
$\boldsymbol{T}$, let $R_t(\boldsymbol{T})$ denote the reverse reachable set of
$t$ in its live-edge graph.  The target belongs to
$L_{\boldsymbol{T}}(X)$ exactly when at least one seed in $X$ can reach $t$,
or equivalently when $X\cap R_t(\boldsymbol{T})\neq\varnothing$.  Hence
Lemma~\ref{lem:simulation} gives
\begin{equation}
	\label{eq:ptx}
p_t(X) = \sum_{\boldsymbol{T}=(T_v)_{v\in V}}
\left(\prod_{v\in V}q_v(T_v)\right)
\1\{X\cap R_t(\boldsymbol{T}) \neq \varnothing\}.
\end{equation}

\section{Alternating Differences as Reachability Cylinders}
\label{sec:cylinders}

Fix a live-edge graph and a target $t$. 
For simplicity, let $R_t$ be the
reverse-reachable set of $t$, including $t$ itself.  The target activates from
seeds $X$ exactly when $X\cap R_t\neq\varnothing$.
Define a set function $h_{R_t}$ as $h_{R_t}(X) = \1\{X\cap R_t\neq\varnothing\}$ for every $X \subseteq V$.

The term \emph{cylinder} comes from the language of product spaces.  Represent
a subset $R\subseteq V$ by its Boolean membership vector
$(\1\{u\in R\})_{u\in V}\in\{0,1\}^V$.  A cylinder event fixes the values of
some coordinates of this vector while leaving all other coordinates
unrestricted.  In particular, the event
$\{A\subseteq R_t,\ S\cap R_t=\varnothing\}$ fixes the reachability coordinate
to one for every node in $A$ and to zero for every node in $S$, while imposing
no condition on the remaining nodes; we therefore call it a \emph{reachability
cylinder}.  The cylinder identity below says that the signed alternating
difference of the target-activation indicator is exactly the indicator of this
event.  When this event is pulled back to the underlying live-edge choices, its
edge constraints need not themselves be a simple product event; the decision
tree in Section~\ref{sec:decision-tree} will partition it into manageable
product rectangles.

\begin{lemma}[Cylinder Identity]
\label{lem:cylinder}
For nonempty $A$ disjoint from $S$,
\begin{equation}
  (-1)^{|A|+1}\Delta_A
  h_{R_t}(S)
  =\1\{A\subseteq R_t,\ S\cap R_t=\varnothing\}.
  \label{eq:cylinderidentity}
\end{equation}
\end{lemma}

\begin{proof}
If $S\cap R_t\neq\varnothing$, every term in the alternating sum has value one
and the sum is zero.  
Now consider $S\cap R_t = \varnothing$ while some $a\in A$ lies outside $R_t$.
In this case, we pair every $B\subseteq A \setminus \{a\}$ with $B \cup \{a\}$.
We have $h_{R_t}(S\cup B) = h_{R_t}(S\cup B \cup \{a\})$ because adding $a$ cannot help activating $t$ ($a$ cannot reach $t$ in the fixed 
	live-edge graph).
But these two terms have opposite signs since $S\cup B \cup \{a\}$ has one more item than $S\cup B$.
Therefore, these two terms cancel in the alternating sum.
This means that when $S\cap R_t\neq\varnothing$ or $A \not\subseteq R_t$, the alternating sum is zero, matching the right-hand side.

Finally, if $S\cap R_t=\varnothing$ and every
$a\in A$ lies in $R_t$, the empty subset of $A$ contributes zero and all
nonempty subsets contribute one; their signed sum is one after multiplication
by $(-1)^{|A|+1}$.
\end{proof}

With the above lemma, we can do the following manipulations.
First, we substitute the signed triggering representation
\eqref{eq:ptx} into the definition of $\Delta_Ap_t(S)$ and interchange the two
finite sums: the sum over triggering configurations and the alternating sum
over subsets of $A$.  For each fixed triggering configuration
$\boldsymbol{T}$, the live-edge graph and its reverse reachable set
$R_t(\boldsymbol{T})$ are fixed.  Lemma~\ref{lem:cylinder} then says that this
configuration contributes its signed weight
$\prod_{v\in V}q_v(T_v)$ exactly when every node in $A$ can reach $t$ and no
node in $S$ can reach $t$; otherwise its contribution to the alternating
difference is zero.  Consequently,
$(-1)^{|A|+1}\Delta_Ap_t(S)$ is the sum of the signed weights of precisely the
configurations in the following event.  Thus, proving the desired inequality for $\AD$-$k$
is equivalent to proving that the product signed measure assigns nonnegative
total mass to
\begin{equation}
  \mathcal{C}(A,S;t)
  :=\{\text{every }a\in A\text{ reaches }t,
      \text{ and no }s\in S\text{ reaches }t\}.
  \label{eq:cylinderevent}
\end{equation}

\section{A Bounded-positive Edge-query Decision Tree}
\label{sec:decision-tree}

As discussed before, the triggering sets $(T_v)_{v \in V}$ form a live-edge graph that determines the diffusion process.
In this section, we apply the query view on the live-edge graph.
Instead of fixing all triggering sets at once, we query each edge $(u,v)\in E$ to check if it is present or absent in the live-edge graph.
At head $v$ of edge $(u,v)$, write $P_v$ for the set of tails fixed present and $Z_v$ for the
set fixed absent.  
Every unqueried membership bit is free.  
These queries form a binary decision tree.
Consequently the
leaf of the decision tree is the product rectangle
\begin{equation}
  \prod_{v\in V}[P_v,N_v\setminus Z_v].
  \label{eq:leafbox}
\end{equation}
This remains true for an adaptive tree: every completion consistent with a
transcript follows the same query path.

\begin{lemma}[Bounded-positive Partition]
\label{lem:partition}
For every nonempty $A$ disjoint from $S$ and every target $t$, the cylinder
$\mathcal{C}(A,S;t)$ has an edge-query decision tree such that every accepting
leaf satisfies
\begin{equation}
  |P_v|\leq |A|
  \qquad\text{for every }v\in V.
  \label{eq:leafbound}
\end{equation}
Equivalently, the cylinder admits a disjoint partition into product Boolean
intervals whose local lower endpoints have size at most $|A|$.
\end{lemma}

\begin{proof}
Fix deterministic orders of vertices and potential edges.

\smallskip
\noindent\textbf{Phase I: expose the forbidden closure.}
Start a simultaneous forward search with discovered set $D=S$.  Query only
edges $u\to v$ with $u\in D$ and $v\notin D$.  When a queried edge is present,
add its head $v$ to $D$; skip all later edges whose head is already discovered.
Continue until no unqueried boundary edge remains.  Then
$D=\Reach(S)$ and $D$ is forward-closed under live edges.  Reject if $t\in D$.

Every present edge queried in Phase I discovers a new head.  Hence Phase I
fixes at most one present incoming edge per vertex, and every such head belongs
to the final $D$.  When $S=\varnothing$, we have $D=\varnothing$ and this phase
makes no query.

\smallskip
\noindent\textbf{Phase II: certify the required paths.}
Reject if any $a\in A$ belongs to $D$.  For each $a\in A\setminus\{t\}$, run a
forward search from $a$ to $t$ in the induced live-edge graph on $V\setminus D$.
Reuse earlier query answers, skip a head once it has been discovered in the
current search, and reject if a search is exhausted before finding $t$.
Accept after all searches succeed.

One source search fixes at most one new present incoming edge at any head.
Across $|A\setminus\{t\}|$ searches, Phase II fixes at most
$|A\setminus\{t\}|$ such edges per head, and every Phase-II head lies outside
$D$.  Phase-I positive heads lie inside $D$, so the phase bounds never add.
Therefore
\begin{equation}
  |P_v|\leq\max\{1,|A\setminus\{t\}|\}\leq |A|,
\end{equation}
where the last inequality also covers $A=\{t\}$.

It remains to prove recognition.  If the tree accepts, Phase I shows that no
member of $S$ reaches $t$, and Phase II shows that every $a\in A$ reaches $t$;
hence \eqref{eq:cylinderevent} holds.  Conversely, suppose
\eqref{eq:cylinderevent} holds.  Phase I returns $D=\Reach(S)$, with $t\notin D$.
No $a\in A$ belongs to $D$, since $S\to a$ and $a\to t$ would imply
$S\to t$.  Moreover an $a$-to-$t$ path cannot enter $D$: because $D$ is
forward-closed, it could not later leave to reach $t\notin D$.  Thus every
required path exists entirely inside $V\setminus D$, every Phase-II search
succeeds, and the tree accepts.  Its accepting leaves are disjoint by
construction, so they partition \eqref{eq:cylinderevent}.
\end{proof}

Intuitively, Lemma~\ref{lem:partition} converts the global reachability event
$\mathcal{C}(A,S;t)$ into disjoint local certificates that can be controlled
by the $\AD$-$k$ assumptions.  The decision tree first determines the region
$D=\Reach(S)$ generated by the forbidden base seeds, and then certifies the
required paths from the nodes in $A$ to $t$ entirely outside $D$.  Since the
present edges exposed in the first phase have heads inside $D$, while those
exposed in the second phase have heads outside $D$, the two contributions do
not accumulate at the same node; consequently, no local triggering set is
forced to contain more than $|A|$ queried-present elements.

This lemma is the central combinatorial bridge in the proof.  The signed
triggering representation and Lemma~\ref{lem:cylinder} identify the desired
global difference with the signed mass of $\mathcal{C}(A,S;t)$, but that mass
is not automatically nonnegative because individual triggering-set weights
may be negative.  Lemma~\ref{lem:partition} partitions the cylinder into
product Boolean intervals whose lower endpoints have size at most
$|A|\leq k$; Lemma~\ref{lem:interval} then guarantees that every local interval
factor, and hence every product box, has nonnegative signed mass.

\section{Proof of the Theorem}
\label{sec:main-proof}

\begin{proof}[Proof of Theorem~\ref{thm:main}]
Fix a target node $t\in V$, and let $A,S\subseteq V$ be arbitrary disjoint
sets with $1\leq |A|\leq k$.  A complete triggering-set configuration is
$\boldsymbol{T}=(T_v)_{v\in V}$.  It induces a fixed live-edge graph, and its
signed weight is
\[
  w(\boldsymbol{T})=\prod_{v\in V}q_v(T_v).
\]
Let $R_t(\boldsymbol{T})$ be the reverse reachable set of $t$ in this
live-edge graph, and define
\[
  h_{\boldsymbol{T}}(X)
  =\1\{X\cap R_t(\boldsymbol{T})\neq\varnothing\}.
\]
Thus, for the fixed live-edge graph induced by $\boldsymbol{T}$,
$h_{\boldsymbol{T}}(X)$ indicates whether the seed set $X$ activates $t$.
Equation~\eqref{eq:ptx}, which follows from the signed-trigger simulation in
Lemma~\ref{lem:simulation}, can be written as
\begin{equation}
  p_t(X)=\sum_{\boldsymbol{T}}w(\boldsymbol{T})
  h_{\boldsymbol{T}}(X).
  \label{eq:pt-signed-sum}
\end{equation}

We now apply the difference operator $\Delta_A$ to both sides.  By its
definition in 
\eqref{eq:difference} and the linearity of finite sums,
\begin{align}
  \Delta_Ap_t(S)
  &=\sum_{B\subseteq A}(-1)^{|A|-|B|}p_t(S\cup B) \notag\\
  &=\sum_{B\subseteq A}(-1)^{|A|-|B|}
    \sum_{\boldsymbol{T}}w(\boldsymbol{T})
    h_{\boldsymbol{T}}(S\cup B) \notag\\
  &=\sum_{\boldsymbol{T}}w(\boldsymbol{T})
    \sum_{B\subseteq A}(-1)^{|A|-|B|}
    h_{\boldsymbol{T}}(S\cup B) \notag\\
  &=\sum_{\boldsymbol{T}}w(\boldsymbol{T})
    \Delta_A h_{\boldsymbol{T}}(S).
  \label{eq:delta-pt-signed-sum}
\end{align}
For each fixed configuration $\boldsymbol{T}$, Lemma~\ref{lem:cylinder}
applies to its fixed live-edge graph.  Therefore,
\begin{align}
  (-1)^{|A|+1}\Delta_Ap_t(S)
  &=\sum_{\boldsymbol{T}}w(\boldsymbol{T})
    \1\{A\subseteq R_t(\boldsymbol{T}),\,
          S\cap R_t(\boldsymbol{T})=\varnothing\} \notag\\
  &=\sum_{\boldsymbol{T}\in\mathcal{C}(A,S;t)}
    \prod_{v\in V}q_v(T_v).
  \label{eq:delta-as-cylinder-mass}
\end{align}
In other words, the signed difference on the left-hand side of
\eqref{eq:delta-as-cylinder-mass} is exactly the product signed mass of the
reachability cylinder $\mathcal{C}(A,S;t)$ defined in
\eqref{eq:cylinderevent}.  Notice that this equality alone does not yet imply
nonnegativity, because the atomic values $q_v(T_v)$ may be negative.

We next use the bounded-positive partition from Lemma~\ref{lem:partition}.
Let $\mathcal{L}_{\mathrm{acc}}$ be the set of accepting leaves of its
edge-query decision tree.  These leaves give the {\em disjoint decomposition}
\begin{equation}
  \mathcal{C}(A,S;t)
  =\bigcup_{\ell\in\mathcal{L}_{\mathrm{acc}}}
    \prod_{v\in V}
    [P_v^{\ell},N_v\setminus Z_v^{\ell}],
  \label{eq:cylinder-leaf-partition}
\end{equation}
where Lemma~\ref{lem:partition} guarantees
\begin{equation}
  |P_v^{\ell}|\leq |A|\leq k
  \qquad
  \text{for every }v\in V
  \text{ and }\ell\in\mathcal{L}_{\mathrm{acc}}.
  \label{eq:accepting-leaf-bound}
\end{equation}
Because the signed measure is the product of the local signed measures and the
union in \eqref{eq:cylinder-leaf-partition} is disjoint, we obtain
\begin{align}
  &\sum_{\boldsymbol{T}\in\mathcal{C}(A,S;t)}
    \prod_{v\in V}q_v(T_v) \notag\\
  &\quad =
    \sum_{\ell\in\mathcal{L}_{\mathrm{acc}}}
    \prod_{v\in V}
    \left(
      \sum_{P_v^{\ell}\subseteq T_v
            \subseteq N_v\setminus Z_v^{\ell}}
      q_v(T_v)
    \right).
  \label{eq:cylinder-leaf-masses}
\end{align}

By Lemma~\ref{lem:interval}, each local factor in
\eqref{eq:cylinder-leaf-masses} equals
\begin{equation}
  \sum_{P_v^{\ell}\subseteq T_v
        \subseteq N_v\setminus Z_v^{\ell}}q_v(T_v)
  =
  \begin{cases}
    1-f_v(Z_v^{\ell}),
      &P_v^{\ell}=\varnothing,\\[2mm]
    (-1)^{|P_v^{\ell}|+1}
      \Delta_{P_v^{\ell}}f_v(Z_v^{\ell}),
      &P_v^{\ell}\neq\varnothing.
  \end{cases}
  \label{eq:leafmass}
\end{equation}
The first case in \eqref{eq:leafmass} is nonnegative because
$f_v(Z_v^{\ell})\leq1$.  In the second case,
\eqref{eq:accepting-leaf-bound} gives
$1\leq|P_v^{\ell}|\leq k$, so the local $\AD$-$k$ assumption on $f_v$
implies
\[
  (-1)^{|P_v^{\ell}|+1}
  \Delta_{P_v^{\ell}}f_v(Z_v^{\ell})\geq0.
\]
Hence every local factor in \eqref{eq:cylinder-leaf-masses} is nonnegative.
Every product over $v$ is therefore nonnegative, and summing over the accepting
leaves gives
\begin{equation}
  \sum_{\boldsymbol{T}\in\mathcal{C}(A,S;t)}
  \prod_{v\in V}q_v(T_v)\geq0.
  \label{eq:cylinder-mass-nonnegative}
\end{equation}
Combining \eqref{eq:delta-as-cylinder-mass} and
\eqref{eq:cylinder-mass-nonnegative} yields
\begin{equation}
  (-1)^{|A|+1}\Delta_Ap_t(S)\geq0.
  \label{eq:pt-adk-final}
\end{equation}

The target $t$ and the disjoint sets $A,S$ were arbitrary, subject only to
$1\leq|A|\leq k$.  Equation~\eqref{eq:pt-adk-final} is therefore exactly the
definition \eqref{eq:adk} of $p_t$ being $\AD$-$k$.

Finally, the influence spread is
$\sigma(X)=\sum_{t\in V}p_t(X)$.  By linearity of the difference operator,
\[
  (-1)^{|A|+1}\Delta_A\sigma(S)
  =\sum_{t\in V}(-1)^{|A|+1}\Delta_Ap_t(S)\geq0.
\]
Thus $\sigma$ is also $\AD$-$k$, proving the theorem.
\end{proof}

%
%
%
%
%

\section{Conclusion}
\label{sec:conclude}

This paper provides a complete proof of the local-to-global $\AD$-$k$ conjecture originally formulated
	by Chen et al.~\cite{adkConjecture}.
It utilizes several techniques, including M\"obius inversion, reverse reachable sets, and decision-tree partitioning.  A central step
	extends the triggering-set distribution to a signed measure, allowing individual triggering sets to have
	negative weights.
Individual negative weights do not pose a problem because the decision-tree partition guarantees that only aggregate weights of order at most $k$ are used,
	and the local $\AD$-$k$ property guarantees that the total weight in the resulting alternating aggregate is nonnegative.
	
The analysis not only resolves an open conjecture but also offers a fresh perspective on general threshold models with local $\AD$-$k$ properties---
	such models can still be represented through reachability in live-edge graphs generated from the triggering sets, but these triggering sets may have negative weights.
Note that the triggering model is equivalent to the $\AD$-$\infty$ model and its triggering sets form the basis of the reverse influence sampling
	approach that leads to near-linear-time algorithms~\cite{borgs2014rrset,tang2014newrrset,tang2015rrset}.
Therefore, investigating whether the extended triggering sets with negative weights have algorithmic implications is an interesting direction worth pursuing further.

\bibliographystyle{elsarticle-num}
\bibliography{ad_k}	

\begin{thebibliography}{10}
\expandafter\ifx\csname url\endcsname\relax
  \def\url#1{\texttt{#1}}\fi
\expandafter\ifx\csname urlprefix\endcsname\relax\def\urlprefix{URL }\fi
\expandafter\ifx\csname href\endcsname\relax
  \def\href#1#2{#2} \def\path#1{#1}\fi

\bibitem{adkConjecture}
W.~Chen, Q.~Li, X.~Shan, X.~Sun, J.~Zhang, Higher order monotonicity and
  submodularity of influence in social networks: {F}rom local to global,
  Information and Computation 285 (2022) 104864.

\bibitem{kempe2003infmax}
D.~Kempe, J.~Kleinberg, {\'E}.~Tardos, Maximizing the spread of influence
  through a social network, in: KDD, ACM, 2003, pp. 137--146.

\bibitem{domingos2001mining}
P.~Domingos, M.~Richardson, Mining the network value of customers, in: KDD,
  ACM, 2001, pp. 57--66.

\bibitem{richardson2002mining}
M.~Richardson, P.~Domingos, Mining knowledge-sharing sites for viral marketing,
  in: KDD, ACM, 2002, pp. 61--70.

\bibitem{Leskovec2007costeffective}
J.~Leskovec, A.~Krause, C.~Guestrin, C.~Faloutsos, J.~VanBriesen, N.~Glance,
  Cost-effective outbreak detection in networks, in: KDD, ACM, 2007, pp.
  420--429.

\bibitem{Budak2011misinformation}
C.~Budak, D.~Agrawal, A.~El~Abbadi, Limiting the spread of misinformation in
  social networks, in: Proceedings of the 20th International Conference on
  World Wide Web, 2011, pp. 665--674.
\newblock \href {https://doi.org/10.1145/1963405.1963499}
  {\path{doi:10.1145/1963405.1963499}}.

\bibitem{chen2013information}
W.~Chen, L.~V.~S. Lakshmanan, C.~Castillo, Information and influence
  propagation in social networks, Morgan \& Claypool Publishers, 2013.

\bibitem{Li2018survey}
Y.~Li, J.~Fan, Y.~Wang, K.-L. Tan, Influence maximization on social graphs: A
  survey, IEEE Transactions on Knowledge and Data Engineering 30~(10) (2018)
  1852--1872.
\newblock \href {https://doi.org/10.1109/TKDE.2018.2807843}
  {\path{doi:10.1109/TKDE.2018.2807843}}.

\bibitem{ChenWY09efficientinfluence}
W.~Chen, Y.~Wang, S.~Yang, Efficient influence maximization in social networks,
  in: KDD, ACM, 2009, pp. 199--208.

\bibitem{chen2010sharpphard}
W.~Chen, C.~Wang, Y.~Wang, Scalable influence maximization for prevalent viral
  marketing in large-scale social networks, in: KDD, ACM, 2010, pp. 1029--1038.

\bibitem{borgs2014rrset}
C.~Borgs, M.~Brautbar, J.~Chayes, B.~Lucier, Maximizing social influence in
  nearly optimal time, in: SODA, ACM-SIAM, 2014, pp. 946--957.

\bibitem{tang2014newrrset}
Y.~Tang, X.~Xiao, Y.~Shi, Influence maximization: near-optimal time complexity
  meets practical efficiency, in: SIGMOD, ACM, 2014, pp. 946--957.

\bibitem{tang2015rrset}
Y.~Tang, Y.~Shi, X.~Xiao, Influence maximization in near-linear time: A
  martingale approach, in: SIGMOD, ACM, 2015, pp. 1539--1554.

\bibitem{Mossel2010local2global}
E.~Mossel, S.~Roch, Submodularity of influence in social networks: From local
  to global, SIAM J. Comput. 39~(6) (2010) 2176--2188.

\bibitem{Nemhauser1978submodular}
G.~L. Nemhauser, L.~A. Wolsey, M.~L. Fisher, An analysis of approximations for
  maximizing submodular set functions---i, Mathematical Programming 14~(1)
  (1978) 265--294.
\newblock \href {https://doi.org/10.1007/BF01588971}
  {\path{doi:10.1007/BF01588971}}.

\bibitem{Sethpan2005booleanfunction}
S.~Foldes, P.~L. Hammer, Submodularity, supermodularity, and higher-order
  monotonicities of pseudo-boolean functions, Mathematics of Operations
  Research 30~(2) (2005) 453--461.

\bibitem{Grabisch2000mobius}
M.~Grabisch, J.-L. Marichal, M.~Roubens, Equivalent representations of set
  functions, Math. Oper. Res. 25 (2000) 157--178.

\bibitem{Ressel2023compounding}
P.~Ressel, On the compounding of higher order monotonic pseudo-boolean
  functions, Positivity 27~(1) (2023) 3.
\newblock \href {https://doi.org/10.1007/s11117-022-00957-3}
  {\path{doi:10.1007/s11117-022-00957-3}}.

\end{thebibliography}

\end{document}